\documentclass[11pt,reqno]{amsart}

\usepackage{amscd,amssymb,amsmath,amsthm}
\usepackage[arrow,matrix]{xy}
\usepackage{graphicx}
\usepackage{color}
\usepackage{cite}
\newtheorem{thm}{Theorem}
\newtheorem{defn}{Definition}
\newtheorem{lemma}{Lemma}
\newtheorem{pro}{Proposition}
\newtheorem{rk}{Remark}

\begin{document}
\title[Gibbs measures for the Potts-SOS model]{Extremality of translation-invariant Gibbs measures for the Potts-SOS model on the Cayley tree}

\author{M. M. Rahmatullaev, M. A. Rasulova }

 \address{M. \ M. \ Rahmatullaev \\ Uzbekistan Academy of Sciences V.I.Romanovskiy Institute of Mathematics, Namangan, Uzbekistan.}
 \email{mrahmatullaev@rambler.ru}

 \address{M.\ A.\ Rasulova \\ Namangan State University, Namangan, Uzbekistan.}
 \email{m$\_$rasulova$\_$a@rambler.ru}

\begin{abstract}
In this paper, we consider the Potts-SOS model where the spin takes values in the set $\{0, 1, 2\}$ on the Cayley tree of order two. We describe all the translation-invariant splitting Gibbs measures
for this model in some conditions. Moreover, we investigate whether these Gibbs measures are extremal or non-extremal in the set of all Gibbs measures.
\end{abstract} \maketitle


\section{Introduction}

One of the central problems in the theory of Gibbs measures (GMs)
is to describe infinite-volume (or limiting) GMs corresponding to
a given Hamiltonian. The existence of such measures for a wide
class of Hamiltonians was established in the ground-breaking work
of Dobrushin (see, e.g., \cite{1}). However, a complete analysis
of the set of limiting GMs for a specific Hamiltonian is often a
difficult problem.

In this paper, we consider the Potts-SOS model, with spin values
$0, 1, 2$ on the Cayley tree (CT). Models on a CT were discussed
in Refs. \cite{2} and \cite{4}-\cite{6}. A classical example of
such a model is the Ising model, with two values of spin $-1$ and $1$.
It was considered in Refs. \cite{14}, \cite{2}, \cite{6}, \cite{16}, \cite{17} and became a focus of active research in the
first half of the 90s and afterwards; see Refs.
\cite{14}, \cite{7}-\cite{13}.

In \cite{18} all translation-invariant splitting Gibbs measures (TISGMs) for the
Potts model on the CT are described. In \cite{19}, \cite{20}
periodic Gibbs measures, in \cite{21}-\cite{23} weakly periodic
Gibbs measures for the Potts model are studied.

In \cite{26}, \cite{27} translation-invariant and periodic Gibbs
measures for the SOS model on the CT are studied.

Model considered in this paper (Potts-SOS model) is
generalization of the Potts and SOS (solid-on-solid) models. In
\cite{15} some translation-invariant Gibbs measures for the
Potts-SOS model on the CT are studied. Periodic Gibbs measures are
studied for the Potts-SOS model on the CT in \cite{25}. In this
paper we will study all the TISGMs for this model under some conditions. Next we investigate whether these Gibbs measures are extremal or non-extremal in the set of all Gibbs measures.

\section{Main definitions and known facts}

The Cayley tree $\Gamma^k$ (See \cite{14}) of order $ k\geq 1 $ is
an infinite tree, i.e., a graph without cycles, from each vertex
of which exactly $ k+1 $ edges issue. Let $\Gamma^k=(V, L, i)$ ,
where $V$ is the set of vertices of $ \Gamma^k$, $L$ is the set of
edges of $ \Gamma^k$ and $i$ is the incidence function associating
each edge $l\in L$ with its endpoints $x,y\in V$. If
$i(l)=\{x,y\}$, then $x$ and $y$ are called {\it nearest
neighboring vertices}, and we write $l=<x,y>$.

The distance $d(x,y), x,y\in V$ on the Cayley tree is defined by the formula
$$d(x,y)=\min\{ d | \exists x=x_0,x_1,...,x_{d-1},x_d=y\in V \ \
\mbox{such that}  \ \ <x_0,x_1>,...,<x_{d-1},x_d> \}.$$

For the fixed $x^0\in V$ we set $ W_n=\{x\in V\ \ |\ \
d(x,x^0)=n\},$
\begin{equation}
V_n=\{x\in V\ \ | \ \  d(x,x^0)\leq n\}, \ \ L_n=\{l=<x,y>\in L \
\ |\ \  x,y\in V_n\}.
\end{equation}
 Denote $|x|=d(x,x^0)$, $x\in V$.

A collection of the pairs $<x,x_1>,...,<x_{d-1},y>$ is called a
{\sl path} from $x$ to  $y$ and we write $\pi(x,y)$ .
 We write $x<y$ if
the path from $x^0$ to $y$ goes through $x$.

It is known (see \cite{14}) that there exists a one-to-one
correspondence between the set  $V$ of vertices  of the Cayley
tree of order $k\geq 1$ and the group $G_{k}$ of the free products
of $k+1$ cyclic  groups $\{e, a_i\}$, $i=1,...,k+1$ of the second
order (i.e. $a^2_i=e$, $a^{-1}_i=a_i$) with generators $a_1,
a_2,..., a_{k+1}$, see Figure \ref{cayley}.

\begin{figure}[h!]
    \begin{center}
        \includegraphics[width=14cm]{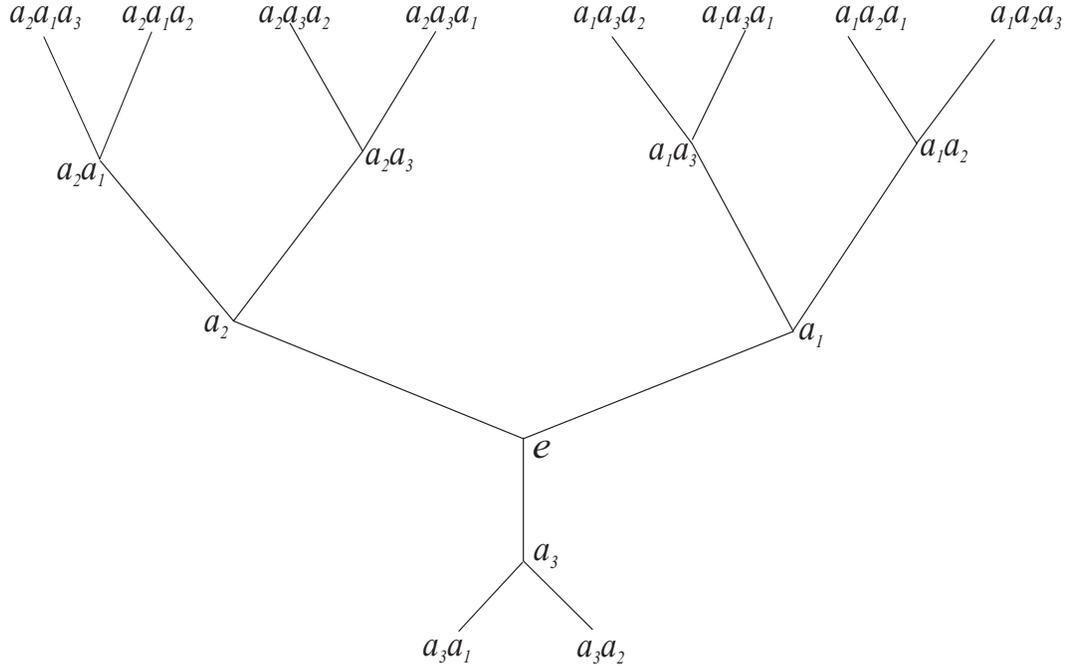}
    \end{center}
    \caption{The Cayley tree $\tau^2$ and elements of the group $G_2$ representation of vertices} \label{cayley}
\end{figure}

Denote the set of "direct successors" of $x\in G_k$ by $S(x)$. Let
$S_1(x)$ be the set of all nearest neighboring vertices of
$x\in G_k,$ i.e. $S_1(x)=\{y\in G_k: <x,y>\}$ and $\{x_{\downarrow}\}
=S_1(x)\setminus S(x)$.

\section{The model and a system vector-valued functional
equations}

Here we shall give main definitions and facts about the model.
Consider model where the spin takes values in the set
$\Phi=\{0,1,2,...,m\}, m\geq 1$. For $A\subseteq V$ a spin {\it
configuration} $\sigma_A$ on $A$ is defined as a function
 $x\in A\to\sigma_A(x)\in\Phi$; the set of all configurations coincides with
$\Omega_A=\Phi^{A}$. Denote $\Omega=\Omega_V$ and
$\sigma=\sigma_V.$

  A configuration that is invariant with respect to all
shifts is called {\it translational-invariant}.

The Hamiltonian  of the Potts-SOS model with nearest-neighbor
interaction has the form
\begin{equation} \label{2}
 H(\sigma)=-J \sum_{<x,y>\in L}|\sigma(x)-\sigma(y)|-
J_p \sum_{<x,y>\in L}\delta_{\sigma(x)\sigma(y)},
\end{equation}
 where $J, J_p\in R$ are
nonzero coupling constants.

It is known \cite{15} that any SGM of the model (\ref{2}) corresponds to a solution of the following equation:

\begin{equation} \label{6}
h_{x}^{*}=\sum_{y\in S(x)}F\big(h_{y}^{*},m,\theta,r\big)
\end{equation}
where  $x\in V \backslash \{x^{0}\}$,
\begin{equation} \label{7}
\theta=\exp(J\beta),   \ \ r=\exp(J_{p}\beta)
\end{equation}
and also $\beta=1/T$ is the inverse temperature. Here $h_{x}^{*}$
represents the vector
$(h_{0,x}-h_{m,x},h_{1,x}-h_{m,x},...,h_{m-1,x}-h_{m,x})$ and the
vector function $F(.,m,\theta,r):R^{m}\rightarrow R^{m}$ is
defined as follows
$$
F(h,m,\theta,r)=\big(F_{0}(h,m,\theta,r),F_{1}(h,m,\theta,r),...,F_{m-1}(h,m,\theta,r)\big),
$$
where
\begin{equation} \label{8}
F_{i}(h,m,\theta,r)=\ln\frac{\sum_{j=0}^{m-1}\theta^{|i-j|}r^{\delta_{ij}}e^{h_{j}}+\theta^{m-i}r^{\delta_{mi}}}{\sum_{j=0}^{m-1}\theta^{m-j}r^{\delta_{mj}}e^{h_{j}}+r},
\end{equation}
$ h=(h_0, h_1,...,h_{m-1}), i=0,1,2,...,m-1. $

Namely, for any collection of functions satisfying the functional
equation (\ref{6}) there exists a unique splitting Gibbs measure,
the correspondence being one-to-one.

\section {Translation-invariant Gibbs measures}

\begin{defn} For a SGM $\mu$, if $h_{j,x}$ is independent from $\{x:
h_{j,x}\equiv h_j, x\in V, j\in \Phi$\}, $\mu$ is called
\textit{translation-invariant}$\big($\textbf{TI}$\big)$.
\end{defn}

Let $m=2$, that is $\Phi=\{0,1,2\}$. In this case, for the TISGMs (\ref{6}) has the form
$$h=kF(h,\theta,r),$$ where $h=(h_0,h_1)$.
 Introducing the notation $l_0=e^{h_0}, l_1=e^{h_1}$,
we obtain the following the system of equations

\begin{equation} \label{11}
\left\{%
\begin{array}{ll}
    l_0=\big(\frac{rl_0+\theta l_1+\theta^2}{\theta^2l_0+\theta l_1+r}\big)^k,\\[3mm]
    l_1=\big(\frac{\theta l_0+rl_1+\theta}{\theta^2l_0+\theta l_1+r}\big)^k. \\[3mm]
   \end{array}%
\right.
\end{equation}

Let $k=2$. Denote $\sqrt{l_0}=x, \sqrt{l_1}=y$. Then from (\ref{11}) we get

\begin{equation} \label{12}
\left\{%
\begin{array}{ll}
    x=\frac{rx^2+\theta y^2+\theta^2}{\theta^2x^2+\theta y^2+r},\\[3mm]
    y=\frac{\theta x^2+ry^2+\theta}{\theta^2x^2+\theta y^2+r}. \\[3mm]
   \end{array}%
\right.
\end{equation}

After simplifying above the system of equations
(\ref{12}), we have

\begin{equation} \label{13}
\left\{%
\begin{array}{ll}
    \theta^2x^3-rx^2+(\theta y^2+r)x-\theta y^2-\theta^2=0,\\[3mm]
    \theta y^3-ry^2+(\theta^2x^2+r)y-\theta x^2-\theta=0. \\[3mm]
   \end{array}%
\right.
\end{equation}

The system of equations (\ref{13}) can be rewritten as

\begin{equation} \label{14}
\left\{%
\begin{array}{ll}
    (x-1)(\theta^2x^2+\theta^2x+\theta^2-rx+\theta y^2)=0,\\[3mm]
    \theta y^3-ry^2+(\theta^2x^2+r)y-\theta x^2-\theta=0. \\[3mm]
   \end{array}%
\right.
\end{equation}

Obviously, the solutions of (\ref{14}) are the solutions of the
following system of equations
\begin{equation} \label{15}
\left\{%
\begin{array}{ll}
    x-1=0,\\[3mm]
    \theta y^3-ry^2+(\theta^2x^2+r)y-\theta x^2-\theta=0, \\[3mm]
   \end{array}%
\right.
\end{equation}

or the solutions of the following system of equations

\begin{equation} \label{16}
\left\{%
\begin{array}{ll}
    \theta^2x^2+\theta^2x+\theta^2-rx+\theta y^2=0,\\[3mm]
    \theta y^3-ry^2+(\theta^2x^2+r)y-\theta x^2-\theta=0. \\[3mm]
   \end{array}%
\right.
\end{equation}

Let us consider (\ref{15}). Substituting $x=1$ into the second equation of
(\ref{15}) we get

\begin{equation} \label{17}
\theta y^3-ry^2+(\theta^2+r)y-2\theta=0.
\end{equation}

For
\begin{equation}
y=z+\frac{r}{3\theta},
\end{equation}
we reduce (\ref{17}) to the equation
\begin{equation} \label{20}
z^3+\big(\frac{r}{\theta}+\theta-\frac{r^2}{3\theta^2}\big)z+\big(\frac{r}{3}+\frac{r^2}{3\theta^2}-\frac{2r^3}{27\theta^3}-2\big)=0.
\end{equation}

Denote
\begin{equation} \label{21}
p=\frac{r}{\theta}+\theta-\frac{r^2}{3\theta^2},
q=\frac{r}{3}+\frac{r^2}{3\theta^2}-\frac{2r^3}{27\theta^3}-2.
\end{equation}

After solving the equation $p=0$ in terms of $r$, we have the solutions
$r_{1,2}=\frac{3\pm \sqrt{9+12\theta}}{2}\theta$. Since $r>0,
\theta>0$, we get $r_1=\frac{3+\sqrt{9+12\theta}}{2}\theta$. Putting $r_1$ into $q$ in (\ref{21}) and solving the equation $q=0$ in
terms of $\theta$, we have the solution
$\theta_1=3\sqrt[3]{2}(\sqrt[3]{2}-1)$.\\
 Substituting $r_1, \theta_1$ into the equation (\ref{20}) we get the equation $z^3=0$. It follows that the equation (\ref{17}) has one positive root $y=\frac{r_1}{3\theta_1}$.\\

From (\ref{21}), we obtain
\begin{equation} \label{*21}
Q(r,\theta)=(\frac{p}{3})^3+(\frac{q}{2})^2=\frac{1}{27}(-\frac{1}{3}\frac{r^2}{\theta^2}+\frac{r}{\theta}+\theta)^3+\frac{1}{4}(-\frac{2}{27}\frac{r^3}{\theta^3}+\frac{1}{3}\frac{r^2}{\theta^2}+\frac{1}{3}r-2)^2=$$$$=
-\frac{1}{108\theta^4}(r^4+2r^3\theta^2+r^2\theta^4-12r^3\theta-12r^2\theta^3-12\theta^5r-4\theta^7+36\theta^2r^2+36\theta^4r-108\theta^4).
\end{equation}
For $\theta=\theta_1=3\sqrt[3]{2}(\sqrt[3]{2}-1)$ we have
$$Q(r,\theta_1)=\frac{116+73\sqrt[3]{4}+92\sqrt[3]{2}}{34992}\big(-r^2+36(1-2\sqrt[3]{2}+\sqrt[3]{4})r+324(13-4\sqrt[3]{2}-5\sqrt[3]{4})\big)\cdot$$$$\cdot\big(r-18+9\sqrt[3]{4}\big)^2$$

Using Cardano's formula one can prove the following

\begin{lemma} \label{l1} Let $\theta=3\sqrt[3]{2}(\sqrt[3]{2}-1)$.
There exists $r_{c}(\approx 4.221293186)$ such that

 $\bullet$ If $r\in (0, r_{c})$ then the equation (\ref{17}) has one positive solution.

$\bullet$ If $r=r_{c}$ then the equation (\ref{17})
has two positive solutions.

$\bullet$ If $r\in (r_{c}, \infty)$ then the equation (\ref{17})
has three positive solutions.
\end{lemma}

Now we consider (\ref{16}). From (\ref{16}) we get

\begin{equation} \label{22}
x=\frac{\theta y(\theta^2-y+ry-r)}{-\theta^3 y+\theta^2+\theta
ry-r}. \end{equation}

Substituting (\ref{22}) into the first equation of (\ref{16}), we obtain

$$f(y,r,\theta)=\theta^2 (\theta+1)(r^2-2\theta
r+\theta^3-\theta^2+\theta)y^4-\theta(r-\theta^2)(r^2+(\theta^2+1)r-3\theta^2)y^3+$$

\begin{equation} \label{23}
+((\theta+1)r+\theta^3)(r-\theta^2)^2y^2-(r+\theta^2)(r-\theta^2)^2y+\theta(r-\theta^2)^2=0.
\end{equation}

The equation (\ref{23}) can be rewritten as
$$f(y,r,\theta)=(ay^2+by+c)(dy^2+ey+f),$$
where

$$ad=\theta^2(\theta+1)(r^2-2\theta r+\theta^3-\theta^2+\theta),$$
$$ae+bd=-\theta(r-\theta^2)(r^2+(\theta^2+1)r-3\theta^2),$$
$$af+be+cd=((\theta+1)r+\theta^3)(r-\theta^2)^2,$$
$$bf+ce=-(r+\theta^2)(r-\theta^2)^2,$$
$$cf=\theta(r-\theta^2)^2.$$

Let $D_1(r, \theta)=b^2-4ac$ and $D_2(r, \theta)=e^2-4df$.

We denote the following sets
$$B_1=\{(r, \theta)\in\mathbb{R}_{+}^{2}: D_1(r, \theta)>0, D_2(r, \theta)>0\},$$
$$B_2=\{(r, \theta)\in\mathbb{R}_{+}^{2}: D_1(r, \theta)>0,D_2(r, \theta)=0\vee D_1(r, \theta)=0,D_2(r, \theta)>0\},$$
$$B_3=\{(r, \theta)\in\mathbb{R}_{+}^{2}: D_1(r, \theta)=0,D_2(r,
\theta)=0\vee D_1(r, \theta)>0,D_2(r, \theta)<0\vee $$
$$\vee D_1(r, \theta)<0,D_2(r, \theta)>0\},$$
$$B_4=\{(r, \theta)\in\mathbb{R}_{+}^{2}: D_1(r, \theta)=0,D_2(r, \theta)<0\vee D_1(r, \theta)<0,D_2(r, \theta)=0\},$$
$$B_5=\{(r, \theta)\in\mathbb{R}_{+}^{2}: D_1(r, \theta)<0, D_2(r, \theta)<0\}.$$

 Thus we can prove the following

\begin{lemma} \label{l2} Let $\theta=3\sqrt[3]{2}(\sqrt[3]{2}-1)$,
then the following assertions hold

 $\bullet$ If $r\in B_1(r)$ then the equation (\ref{23}) has four
 solutions which are positive.

$\bullet$ If $r\in B_2(r)$ then the equation (\ref{23}) has three
positive solutions.

$\bullet$ If $r\in B_3(r)$ then the equation (\ref{23}) has two
positive solutions.

$\bullet$ If $r\in B_4(r)$ then the equation (\ref{23}) has one positive
solution.

$\bullet$ If $r\in B_5(r)$ then the equation (\ref{23}) has no
solution.
\end{lemma}

With respect to (\ref{21}) and (\ref{*21}) we denote the following sets
$$A_1=\{(r, \theta)\in\mathbb{R}_{+}^{2}:r\leq3\theta^2, Q>0\}\cup\{(r, \theta)\in\mathbb{R}_{+}^{2}:r\leq3\theta^2, p=0, q=0\},$$
$$A_2=\{(r, \theta)\in\mathbb{R}_{+}^{2}:r\leq3\theta^2, Q=0\}\cap\{(r, \theta)\in\mathbb{R}_{+}^{2}:p\neq0 \vee q\neq0\},$$
$A_3=\{(r, \theta)\in\mathbb{R}_{+}^{2}:r\leq3\theta^2, Q<0\},$
$A_4=\{(r, \theta)\in\mathbb{R}_{+}^{2}:r>3\theta^2, Q>0\},$
$$A_5=\{(r, \theta)\in\mathbb{R}_{+}^{2}:r>3\theta^2, Q=0\}\cap\{(r, \theta)\in\mathbb{R}_{+}^{2}:p\neq0 \vee q\neq0\},$$
$$A_6=\{(r, \theta)\in\mathbb{R}_{+}^{2}:r>3\theta^2, Q<0\}.$$

Let $N$ be the number of TISGMs for the Potts-SOS model.
 \begin{thm} \label{t2} Let $k=2, m=2$. The following statements hold for the $N$
\begin{equation}\label{*t2}
 N=\left\{ \begin{array}{ll}
    1, \   \ \mbox{if}  \ (r, \theta)\in A_1,\\[3mm]
    2, \   \ if  \ (r, \theta)\in A_2\cup (A_4\cap B_4)\cup (A_5\cap B_5),\\[3mm]
    3, \   \ if  \ (r, \theta)\in A_3\cup (A_4\cap B_3)\cup (A_5\cap B_4),\\[3mm]
    4, \   \ if  \ (r, \theta)\in (A_4\cap B_2)\cup (A_5\cap B_3)\cup (A_6\cap B_4),\\[3mm]
    5, \   \ if  \ (r, \theta)\in (A_4\cap B_1)\cup (A_5\cap B_2)\cup (A_6\cap B_3),\\[3mm]
    6, \   \ if  \ (r, \theta)\in  (A_5\cap B_1)\cup (A_6\cap B_2),\\[3mm]
    7, \   \ if  \ (r, \theta)\in  A_6\cap B_1.\\[3mm]
 \end{array} \right.\end{equation}
\end{thm}

\begin{proof} We consider the first equation of (\ref{16}). We
write this in the following form
\begin{equation} \label{24}
\theta^2x^2+(\theta^2-r)x+\theta^2=-\theta y^2.
\end{equation}

RHS of (\ref{24}) is negative, thus
\begin{equation} \label{25}
\theta^2x^2+(\theta^2-r)x+\theta^2<0.
\end{equation}

For LHS of (\ref{25}), we calculate its discriminant $D=(\theta^2-r)^2-4\theta^4.$
If the discriminant is positive, then the inequality (\ref{25}) has real solutions. Therefore, we should solve
$$(-r-\theta^2)(3\theta^2-r)>0.$$
Since $-r-\theta^2<0$, it follows that $r>3\theta^2$.

 Inequality (\ref{25}) has positive solution as soon as $\theta^2-r<0$ or $r>\theta^2$. If $r>3\theta^2$, then $r>\theta^2$ also holds. If $r>3\theta^2$, the solutions of the inequality (\ref{25}) belong to
 $$\left(\frac{r-\theta^2-\sqrt{D}}{2\theta^2}, \frac{r-\theta^2+\sqrt{D}}{2\theta^2}\right).$$
 Moreover, (\ref{24}) holds in this interval.

Consequently, if $r>3\theta^2$ then the first equation of
(\ref{16}) has a positive real solution, if $r\leq3\theta^2$ then
the first equation of (\ref{16}) cannot have a positive solution, i.e., any positive real pair $(x, y)$, which is solution of the first equation of
(\ref{16}), does not satisfy  $r\leq3\theta^2$. Then
TISGMs corresponding roots of (\ref{16}) do not exist under condition $r\leq3\theta^2$.

According to the Descartes theorem, the number of positive roots of
equation (\ref{17}) is at least 1 and at most 3.

If $Q>0$, then the equation (\ref{20}) has one positive real root
and two conjugate complex roots; If $Q=0$, the all roots of the
equation (\ref{20}) are positive real and two of them are equal or
if $p=q=0$, then (\ref{20}) has one positive real root (one real
zero of multiplicity three); If $Q<0$, then the equation
(\ref{20}) has three distinct positive real roots. Hence, we can say about the number of TISGMs corresponding positive
roots of the equation (\ref{17}).

From Lemma \ref{l1} and Lemma \ref{l2}, we can see that
$$\left\{(r,\theta)\in R^2:\theta=3\sqrt[3]{2}(\sqrt[3]{2}-1), r\in(r_c, \infty)\cap B_1(r)\right\}\subset A_6\cap B_1.$$
Thus the set $A_6\bigcap B_1$ is not empty, i.e., the number of TISGMs corresponding positive
solutions of (\ref{13}) for the Potts-SOS model is up to seven.\end{proof}

\begin{rk}
Note that Theorem \ref{t2} (for $k=m=2$) generalizes results of \cite{18} and \cite{27}.\end{rk}

If $J=0$, then Potts-SOS model changes to Potts model. In this case Theorem
\ref{t2} can be restated as follows
\begin{thm} \label{t3} Let $k=2, m=2$. The following statements hold for the number $n$ of the TISGMs for the
Potts model

\begin{equation}\label{*t3}
 n=\left\{ \begin{array}{ll}
 1, \, \mbox{if} \, \ r\in(0, 1+2\sqrt{2}),\\[2mm]
 4,\,  \mbox{if} \, \ r=1+2\sqrt{2} \ \mbox{or} \  r=4,\\[2mm]
 7, \, \mbox{if} \, \ r\in(1+2\sqrt{2}, 4)\cup(4, \infty).\\
 \end{array} \right.\end{equation}
\end{thm}

(see \cite{18} for more details).

If $J_p=0$, then Hamiltonian (\ref{2}) of Potts-SOS model changes
to Hamiltonian of SOS model. In this case Theorem
\ref{t2} can be restated as follows
\begin{thm} \label{t4} Let $k=2, m=2$. The following statements are appropriate for the number $n$ of the TISGMs for the
SOS model
\begin{equation}\label{*t4}
 n=\left\{ \begin{array}{ll}
 1, \, \mbox{if} \, \ \theta\in (\theta_{2},\infty),\\[2mm]
 3,\,  \mbox{if} \, \ \theta=\theta_2,\\[2mm]
 5, \, \mbox{if} \, \ \theta\in (\theta_{1}, \theta_{2}),\\[2mm]
 6, \, \mbox{if} \, \ \theta= \theta_{1},\\[2mm]
 7, \, \mbox{if} \, \ \theta\in (0, \theta_1),\\
 \end{array} \right.\end{equation}
where $\theta_1\approx 0.1414$ and $\theta_2\approx 0.2956$.
\end{thm}

(see \cite{27} for more details).

Now we study the extremality of the TISGMs for the Potts-SOS model. In general, a complete analyses of extremality or non-extremality of the TISGMs is a difficult problem. Therefore, we assume $r=\theta^2$.

\begin{lemma}\label{l3} Let $r=\theta^2$. There exists a unique $\theta_c(\approx7.729814)$ such that\\
$\bullet$ If $\theta \in (0,\theta_c)$ then system (\ref{12}) has one positive root.\\
$\bullet$ If $\theta=\theta_c$ then system (\ref{12}) has two positive roots.\\
$\bullet$ If $\theta\in (\theta_c, \infty)$ then system (\ref{12}) has three positive roots.\\
\end{lemma}

\begin{proof}

Substituting $r=\theta^2$ into (\ref{12}) we have 					
\begin{equation} \label{33}
\left\{%
\begin{array}{ll}
    x=1,\\[3mm]
    y=\frac{2+\theta y^2}{2 \theta+y^2}. \\[3mm]
   \end{array}%
\right.
\end{equation}

Simplifying the second equation of (\ref{33}),  we obtain the cubic equation
 \begin{equation} \label{34}
y^3-\theta y^2+2\theta y-2=0.
\end{equation}
		
We calculate its discriminant
 \begin{equation} \label{35}
D=4(\theta^4-10\theta^3+18\theta^2-27).
\end{equation}
Denote $\theta_c\approx 7.729814$. If $D<0$ ($\theta<\theta_c$) the equation (\ref{34}) has one real and two conjugate complex roots. If  $D=0$ ($\theta = \theta_c$) then all roots of equation (\ref{34}) are real, which two of them are equal. If $D>0$ ($\theta>\theta_c$) then the equation  (\ref{34}) has three distinct real roots (see Fig. \ref{y1y2y3grafik}). Obtained real roots are positive due to the Descartes theorem (see \cite{24}).
\end{proof}

\begin{figure}[h!]
    \begin{center}
        \includegraphics[width=10cm]{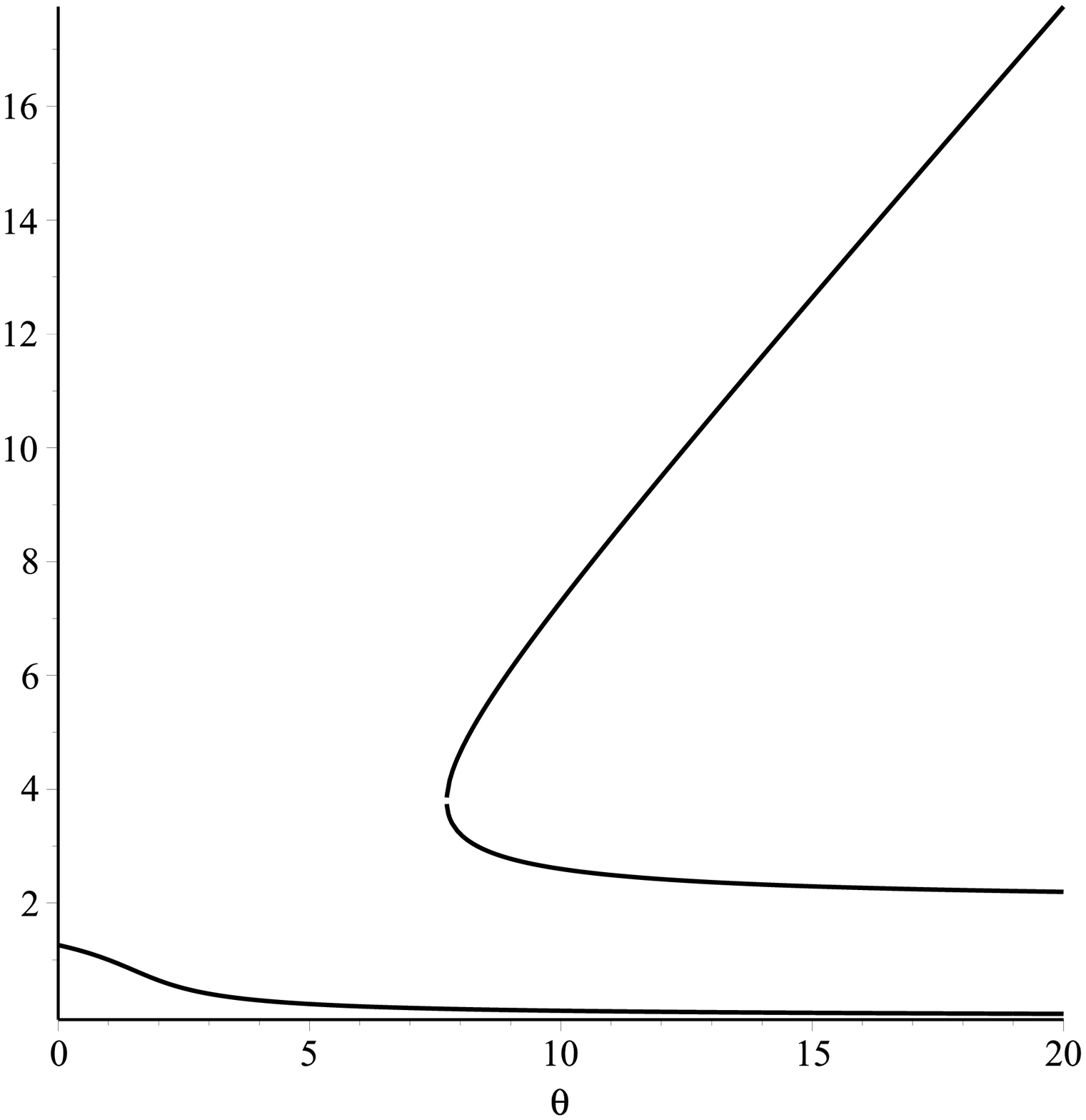}
    \end{center}
    \caption{The graphs of functions $y_i=y_i(\theta), i=1,2,3.$} Lower curve is $y_1$, middle curve is $y_2$, upper curve is $y_3$  \label{y1y2y3grafik}
\end{figure}

Using Lemma \ref{l3}, we have the following

\begin{thm}\label{t5} Let $k=m=2$. If $r=\theta^2$ then the following statements hold for the $N$
 \begin{equation}\label{36}
 N=\left\{ \begin{array}{ll}
    1, \   \ \mbox{if}  \ \theta \in (0, \theta_c),\\[3mm]
    2, \   \ \mbox{if}  \ \theta=\theta_c,\\[3mm]
    3, \   \ \mbox{if}  \ \theta\in (\theta_c, \infty),\\[3mm]
   \end{array} \right.\end{equation}
   where $\theta_c\approx7.729814$.
\end{thm}

\begin{rk}
Note that the Theorem \ref{t5} is a particular case of the Theorem \ref{t2} .
\end{rk}

We denote obtained TISGMs corresponding to $y_i$ in the Theorem \ref{t5} by $\mu_i, i=1, 2, 3$, respectively.

\section{Tree-indexed Markov Chains of TISGMs}
A tree-indexed Markov chain is defined as follows. Suppose we are given with vertices set $V$, a probability measure $\nu$ and a transition matrix $P=(p_{i,j})_{i,j\in \Phi}$ on the single-site space which is here the finite set $\Phi=\{0, 1, ..., m\}$. We can obtain a tree-indexed Markov chain $X: V\rightarrow\Phi$ by choosing $X(x_0)$ according to $\nu$ and choosing $X(v)$, for each vertex $v\neq x^0$, using the transition probabilities given the value of its parent, independently of everything else. See Definition 12.2 in \cite{4} for a detailed definition.

We note that a TISGM corresponding to a vector $v=(x,y)\in R^2$ (which is solution to the system (\ref{12})) is a tree-indexed Markov chain with states $\{0, 1, 2\}$ and transition probabilities matrix:
\begin{equation}\label{51}
 P=\left( \begin{array}{ll}
 \frac{rx^2}{rx^2+\theta y^2+\theta^2} \ \ \frac{\theta y^2}{rx^2+\theta y^2+\theta^2} \ \ \frac{\theta^2}{rx^2+\theta y^2+\theta^2}\\[2mm]
 \frac{\theta x^2}{\theta x^2+ry^2+\theta} \ \ \frac{ry^2}{\theta x^2+ry^2+\theta} \ \ \frac{\theta}{\theta x^2+ry^2+\theta}\\[2mm]
 \frac{\theta^2 x^2}{\theta^2 x^2+\theta y^2+r} \ \ \frac{\theta y^2}{\theta^2 x^2+\theta y^2+r} \ \ \frac{r}{\theta^2 x^2+\theta y^2+r}\\
 \end{array} \right).\end{equation}

Since $(x,y)$ is a solution to the system (\ref{12}) this matrix can be written in the following form
\begin{equation}\label{52}
 P=\frac{1}{Z}\left( \begin{array}{ll}
 rx \ \ \ \frac{\theta y^2}{x} \ \ \ \frac{\theta^2}{x}\\[2mm]
 \frac{\theta x^2}{y} \ \ \ ry \ \ \ \frac{\theta}{y}\\[2mm]
 \theta^2x^2 \ \ \theta y^2 \ \ r\\
 \end{array} \right),\end{equation}
 where $Z=\theta^2x^2+\theta y^2+r$.\\

 Simple calculations show that the matrix (\ref{52}) has three eigenvalues: $1$ and
\begin{equation}\label{53}
 \lambda_{1}(x,y,\theta,r)=\frac{(x+y+1)r-Z+\sqrt{D^*}}{2Z},
 \lambda_{2}(x,y,\theta,r)=\frac{(x+y+1)r-Z-\sqrt{D^*}}{2Z},
\end{equation}
where $\lambda_1$ and $\lambda_2$ are solutions to
\begin{equation}\label{53}
Z^3\lambda^2+(Z-(1+x+y)r)Z^2\lambda+(2\theta^4-\theta^4r-2\theta^2r+r^3)xy=0
\end{equation}
and $D^*=((1+x+y)r-Z)^2-4xyZ^{-1}(2\theta^4-\theta^4r-2\theta^2r+r^3).$
\subsection{Conditions of Non-Extremality}

In this subsection we are going to find the regions of the parameter $\theta$ where the TISGMs $\mu_i, i=1,2,3$ are not extreme in the set of all Gibbs measures (including the non-translation invariant ones).\\
It is known that a sufficient condition (Kesten-Stigum condition) for non-extremality of a Gibbs measure $\mu$ corresponding to the matrix $P$ on a Cayley tree of order $k\geq1$ is that $k\lambda_{\max}^{2}>1$, where $\lambda_{\max}$ is the second largest (in absolute value) eigenvalue of $P$ \cite{28}. We are going to use this condition for TISGMs $\mu_i, i=1,2,3$ in Theorem \ref{t5}. We have all solutions of the system (\ref{12}) in condition $r=\theta^2$ (see Theorem \ref{t5}) and the eigenvalues of the matrix $P$ in the explicit form.\\

Let us denote
$$\lambda_{\max,i}(\theta,r)=\max\{|\lambda_1(x_i,y_i,\theta,r)|, |\lambda_2(x_i,y_i,\theta,r)|\},  i=1,2,3.$$
Using a computer we have
$$ \lambda_{\max,i}(\theta)=\left\{ \begin{array}{ll}
 |\lambda_2(1,y_1,\theta)|, \, \mbox{if} \, \ i=1, \theta<1,\\[2mm]
 |\lambda_1(1,y_1,\theta)|, \, \mbox{if} \, \ i=1, \theta>1,\\[2mm]
 |\lambda_1(1,y_i,\theta)|, \, \mbox{if} \, \ i=2,3.
 \end{array} \right.$$

Denote
$$ \eta_i(\theta)=2\lambda_{\max,i}^2(\theta)-1, i=1,2,3.$$

Let $\theta<\theta_c$. Using the Cardano formula, we solve the equation (\ref{34}). It has one real solution
\begin{equation}
y_1=\frac{1}{3}\left(\theta+\sqrt[3]{\theta^3-9\theta^2+27+1.5\sqrt{-3D}}+\frac{\theta^2-6\theta}{\sqrt[3]{\theta^3-9\theta^2+27+1.5\sqrt{-3D}}}\right),
\end{equation}
where $D$ is defined in (\ref{35}).
In this case, we are aiming to check the Kesten-Stigum condition of the non-extremality of the measure $\mu_1$. To determine the non-extremality interval of  TISGM $\mu_1$, we should check the condition
$$2\lambda^2_{max,1}-1>0.$$

Using a \emph{Maple} program, one can see that the last inequality holds for $\theta\in (0, \theta_1) (\theta_1\approx0.1666993311)$, which implies that the TISGM $\mu_1$ is not-extreme in this interval (see Fig. \ref{eta1}).

\begin{figure}[h!]
    \begin{center}
        \includegraphics[width=14cm]{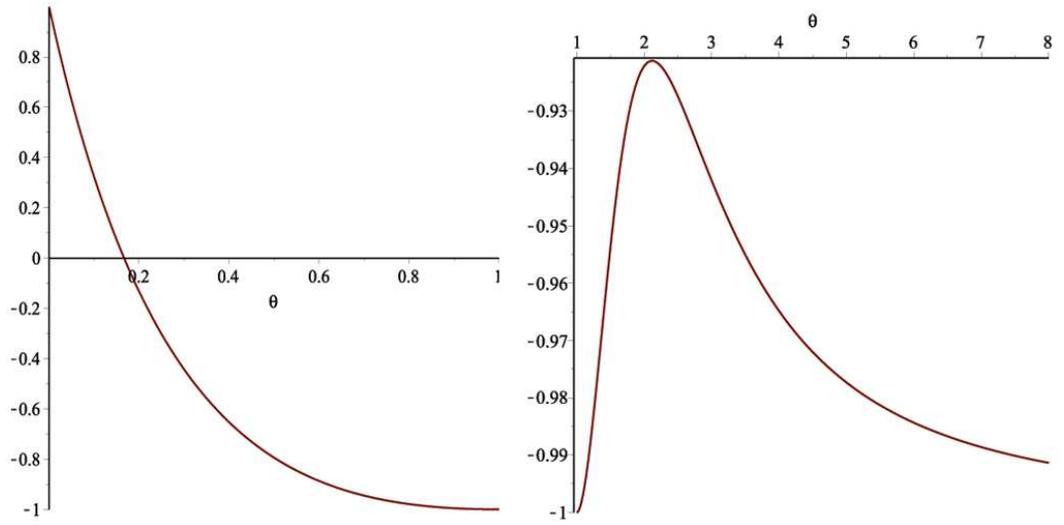}
    \end{center}
    \caption{The graphs of functions $\eta_1(\theta)$ for $\theta\in(0,1)$(left) and for $\theta\in(1,\infty)$(right)} \label{eta1}
\end{figure}
To check that the TISGM $\mu_i, i=2,3 $ are non-extreme, we should solve the following inequality:
$\eta_i(\theta)>0, i=2,3.$ (see Fig. \ref{eta23}).
\begin{figure}[h!]
    \begin{center}
        \includegraphics[width=14cm]{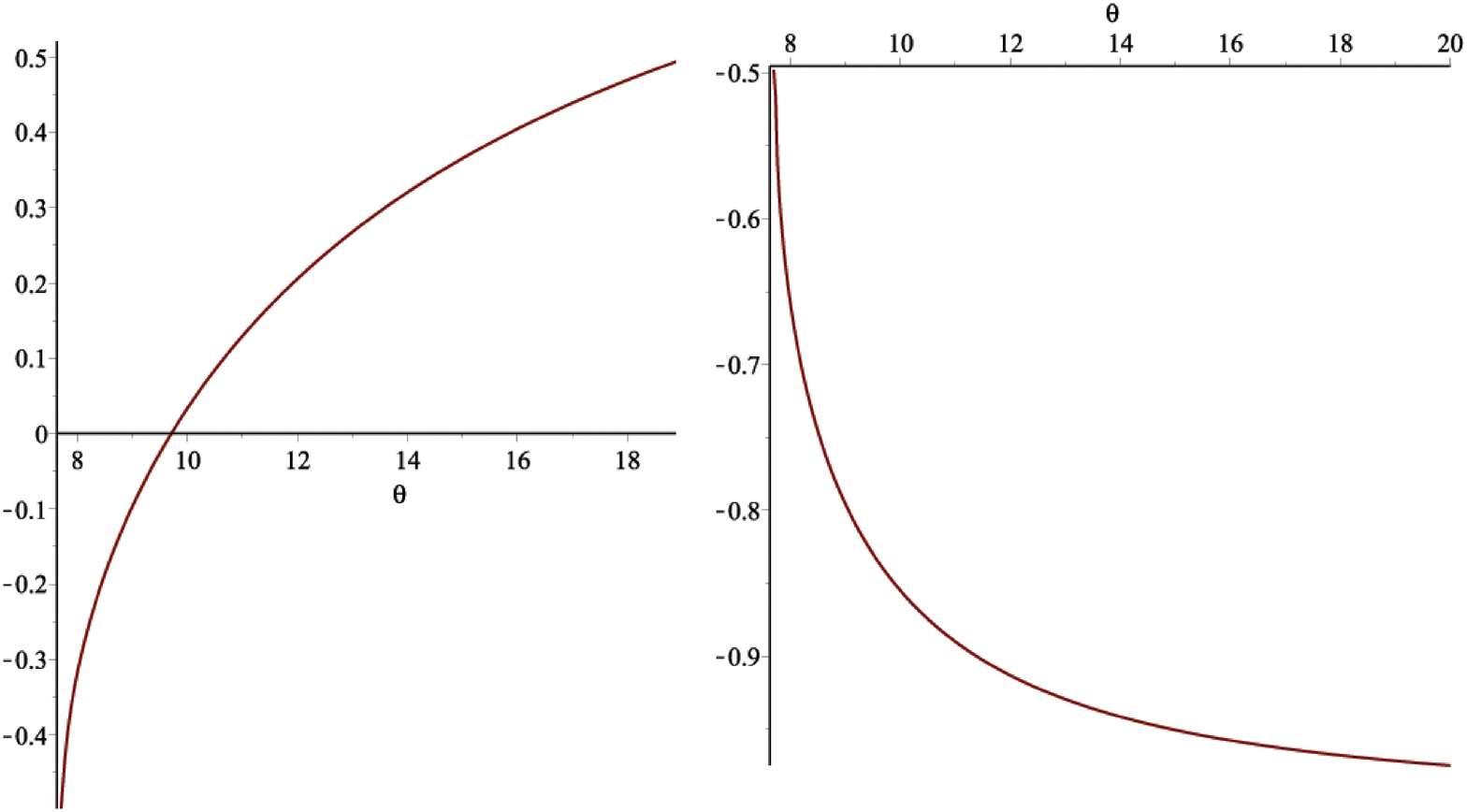}
    \end{center}
    \caption{The graphs of functions $\eta_2(\theta)$ (left) and $\eta_3(\theta)$ (right)} \label{eta23}
\end{figure}

\begin{pro}\label{p1} Let $r=\theta^2$. Then the following statements hold\\
a) There exists $\theta_1(\approx0.1666993311)$ such that the measure $\mu_1$ is non-extreme if $\theta\in (0, \theta_1)$;\\
b) There exists $\theta_2(\approx9.706301628)$ such that the measure $\mu_2$ is non-extreme if $\theta\in(\theta_2, \infty)$.\\
\end{pro}

\subsection{Conditions for Extremality}
In \cite{27}, \cite{29} the key ingredients are two quantities, $\kappa$ and $\gamma$, which bound the rates of percolation of disagreement down and up the tree, respectively.

For two measures $\mu_1$ and $\mu_2$ on $\Omega$, $\|\mu_1-\mu_2\|_x$ denotes the variation distance between the projections of $\mu_1$ and $\mu_2$ onto the spin at $x$, i.e.,
$$\|\mu_1-\mu_2\|_x=\frac{1}{2}\sum_{i=0}^2|\mu_1(\sigma(x)=i)-\mu_2(\sigma(x)=i)|.$$
Let $\eta^{x,s}$ be the configuration $\eta$ with the spin at $x$ set to $s$.
Following \cite{27}, \cite{29} define
$$\kappa\equiv\kappa(\mu)=\sup_{x\in\Gamma^k}\max_{x,s,s^\prime}\|\mu_{\tau_x}^s-\mu_{\tau_x}^{s^\prime}\|_x;$$
$$\gamma\equiv\gamma(\mu)=\sup_{A\subset\Gamma^k}\max\|\mu_{A}^{\eta^{y,s}}-\mu_{A}^{\eta^{y,s^\prime}}\|_x,$$
where the maximum is taken over all boundary conditions $\eta$, all sites $y\in\partial A$, all neighbors $x\in A$ of $y$, and all spins $s, s^\prime\in\{0, 1, 2\}$.

The criterion of extremality of a TISGM is $k\kappa\gamma<1$\cite{27}, \cite{29}.
Note that $\kappa$ has the particularly simple form
$\kappa=\frac{1}{2}\max_{i,j}\sum_{l}|P_{i,l}-P_{j,l}|$
and $\gamma$ is a constant which does not have a clear general formula.

Let $r=\theta^2$. For the solution $(1, y)$, we shall compute $\kappa$
 \begin{equation}
 \kappa=\frac{2\cdot|1-\theta y|+y^2\cdot|\theta-y|}{2y(2\theta+y^2)}.
\end{equation}

For $\theta<1$ from the system (\ref{12}) we get the following inequalities
$$1-\theta y=\frac{\theta(1-\theta^2)y^2}{Z}>0, y-\theta=\frac{2\theta(1-\theta^2)}{Z}>0.$$
Using these inequalities, we obtain

$$ \kappa=\left\{ \begin{array}{ll}
 \frac{y^3-\theta y^2-2\theta y+2}{2y(2\theta+y^2)}, \, \mbox{if} \, \  0<\theta<1,\\[2mm]
 \frac{-y^3+\theta y^2+2\theta y-2}{2y(2\theta+y^2)}, \, \mbox{if} \, \  \theta\geq1.\\[2mm]
 \end{array} \right.$$

 For the solution $(1, y)$, we shall calculate $\gamma$.
 $$\gamma=\max\left\{\|\mu_A^{\eta^{y,0}}-\mu_A^{\eta^{y,1}}\|_x, \|\mu_A^{\eta^{y,0}}-\mu_A^{\eta^{y,2}}\|_x, \|\mu_A^{\eta^{y,1}}-\mu_A^{\eta^{y,2}}\|_x\right\},$$
 where
 $$\|\mu_A^{\eta^{y,0}}-\mu_A^{\eta^{y,1}}\|_x=\frac{1}{2}\sum_{s\in\{0, 1, 2\}}|\mu_A^{\eta^{y,0}}(\sigma(x)=s)-\mu_A^{\eta^{y,1}}(\sigma(x)=s)|=$$
 $$=\frac{1}{2}\left(|P_{0,0}-P_{1,0}|+|P_{0,1}-P_{1,1}|+|P_{0,2}-P_{1,2}|\right)=$$
 $$=\left\{ \begin{array}{ll}
 \frac{y^3-\theta y^2-2\theta y+2}{2y(2\theta+y^2)}, \, \mbox{if} \, \  0<\theta<1,\\[2mm]
 \frac{-y^3+\theta y^2+2\theta y-2}{2y(2\theta+y^2)}, \, \mbox{if} \, \  \theta\geq1,\\[2mm]
 \end{array} \right.$$

$$\|\mu_A^{\eta^{y,0}}-\mu_A^{\eta^{y,2}}\|_x=\frac{1}{2}\sum_{l\in\{0, 1, 2\}} |P_{0,l}-P_{2,l}|=0,$$

$$\|\mu_A^{\eta^{y,1}}-\mu_A^{\eta^{y,2}}\|_x=\frac{1}{2}\sum_{l\in\{0, 1, 2\}} |P_{1,l}-P_{2,l}|=$$
 $$=\left\{ \begin{array}{ll}
 \frac{y^3-\theta y^2-2\theta y+2}{2y(2\theta+y^2)}, \, \mbox{if} \, \  0<\theta<1,\\[2mm]
 \frac{-y^3+\theta y^2+2\theta y-2}{2y(2\theta+y^2)}, \, \mbox{if} \, \  \theta\geq1.\\[2mm]
 \end{array} \right.$$
 Hence, when $0<\theta<1$
$$\gamma=\max\left\{0, \frac{y^3-\theta y^2-2\theta y+2}{2y(2\theta+y^2)}\right\}=\frac{y^3-\theta y^2-2\theta y+2}{2y(2\theta+y^2)},$$
and when $\theta\geq1$
$$\gamma=\max\left\{0, \frac{-y^3+\theta y^2+2\theta y-2}{2y(2\theta+y^2)}\right\}=\frac{-y^3+\theta y^2+2\theta y-2}{2y(2\theta+y^2)}.$$

Now for TISGMs $\mu_i, i=1,2,3$ we want to check the extremality condition $2\kappa\gamma<1$. When $\theta>0$ this condition has the form
$$2\kappa\gamma-1=2\left(\frac{y_i^3-\theta y_i^2-2\theta y_i+2}{2y_i(2\theta+y_i^2)}\right)^2-1<0.$$
We check this condition for the TISGM $\mu_2$.
Denote
$$U_2(\theta)=\frac{(y_2^3-\theta y_2^2-2\theta y_2+2)^2}{2y_2^2(2\theta+y_2^2)^2}-1.$$
The function $U_2(\theta)$ only depends on $\theta$ and has no additional parameters. From its graph one can see the region of $\theta$ where the function is negative. Thus looking on the graph of $U_2(\theta)$ (see Fig. \ref{U2}) completes the arguments.
\begin{figure}[h!]
    \begin{center}
        \includegraphics[width=8cm]{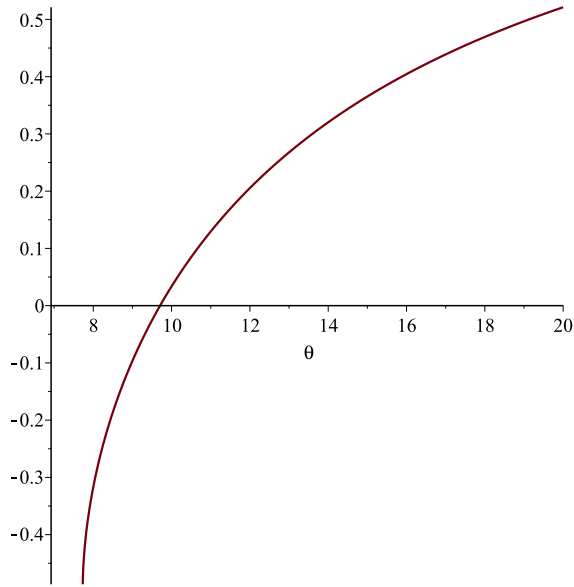}
    \end{center}
    \caption{The graph of function $U_2(\theta)$} \label{U2}
\end{figure}

We check extremality of TISGMs $\mu_1, \mu_3$. Thus consider the following functions
$$U_1(\theta)=\frac{(y_1^3-\theta y_1^2-2\theta y_1+2)^2}{2y_1^2(2\theta+y_1^2)^2}-1,$$
$$U_3(\theta)=\frac{(y_3^3-\theta y_3^2-2\theta y_3+2)^2}{2y_3^2(2\theta+y_3^2)^2}-1.$$
The extremality interval of TISGMs $\mu_1, \mu_3$ are seen from Fig.\ref{U1U3}.

\begin{figure}[h!]
    \begin{center}
        \includegraphics[width=14cm]{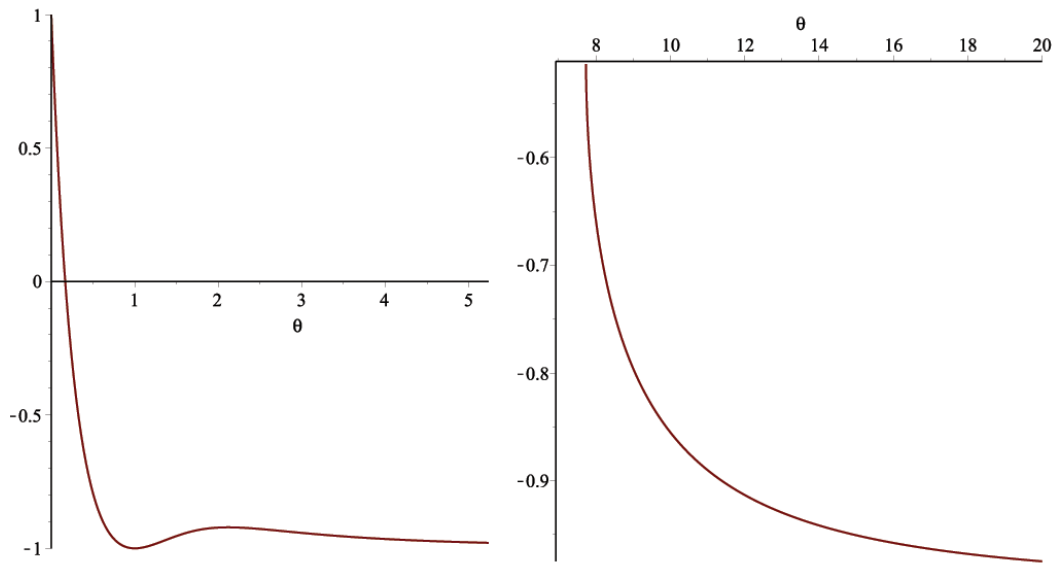}
    \end{center}
    \caption{The graphs of functions $U_1(\theta)$(left) and $U_3(\theta)$(right)} \label{U1U3}
\end{figure}

\begin{pro}\label{p2} Let $r=\theta^2$. Then the following statements hold\\
a) There exists $\theta_1(\approx0.1666993311)$ such that the measure $\mu_1$ is extreme if $\theta\in(\theta_1, \infty)$;\\
b) There are values $\theta^*(\approx7.729813675)$ and $\theta_2(\approx9.706301628)$ such that the measure $\mu_2$ is extreme if  $\theta\in[\theta^*, \theta_2)$;\\
c) The measure $\mu_3$ is extreme (where it exists, that is $\theta\in[\theta^*, \infty)$).\\
\end{pro}

From Proposition \ref{p1} and Proposition \ref{p2} we have the following

\begin{thm} Let $r=\theta^2$. Then the following statements hold\\
a) There exists $\theta_1(\approx0.1666993311)$ such that the measure $\mu_1$ is non-extreme if $\theta\in (0, \theta_1)$  and is extreme if $\theta\in(\theta_1, \infty)$;\\
b) There are values $\theta^*(\approx7.729813675)$ and $\theta_2(\approx9.706301628)$ such that the measure $\mu_2$ is extreme if  $\theta\in[\theta^*, \theta_2)$  and is non-extreme if $\theta\in(\theta_2, \infty)$;\\
c) The measure $\mu_3$ is extreme (where it exists, that is $\theta\in[\theta^*, \infty)$) (see Fig. \ref{y1y2y3grafikox}).\\
\end{thm}

\begin{figure}[h!]
    \begin{center}
        \includegraphics[width=10cm]{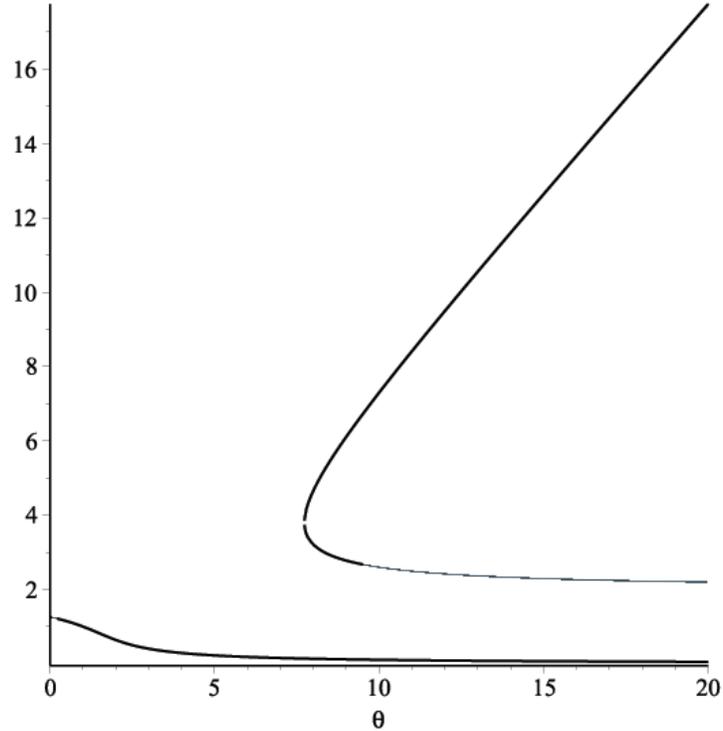}
    \end{center}
    \caption{The graphs of functions $y_i(\theta), i=1,2,3.$ The \textit{bold curves} correspond to regions of the functions where the corresponding TISGM is extreme. The \textit{thin curves} correspond to regions of the functions where the corresponding TISGM is non-extreme} \label{y1y2y3grafikox}
\end{figure}

\section*{ Acknowledgements}
The authors are greatly indebted to Professor U.A. Rozikov for suggesting the problem and for many stimulating conversations.


\begin{thebibliography}{99}

\bibitem{14} U.A. Rozikov,  \textit{Gibbs measures on Cayley trees.} World
scientific, (2013).

\bibitem{1} Ya. G. Sinai, \textit{Theory of Phase Transitions: Rigorous
Results,} Pergamon, (1982).

\bibitem{2} C. Preston, \textit{Gibbs States on Countable Sets,} Cambridge
Univ.Press, (1974).

\bibitem{3} V. A. Malyshev and R. A. Minlos, \textit{Gibbs Random Fields,}
Nauka, (1985).

\bibitem{4} H. O. Georgii, \textit{Gibbs Measures and Phase Transitions, }
Walter de Gruyter, (1988).

\bibitem{24} A. G. Kurosh,  \textit{Nauka}. \textbf{9}, (1968).

\bibitem{5} S. Zachary, \textit{Ann. Probab.} \textbf{11}:4, 894--903, (1983).

\bibitem{6} S. Zachary, \textit{ Stoch. Process. Appl.}
\textbf{20}:2, 247--256, (1985).

\bibitem{7} P. M. Bleher and N. N. Ganikhodjaev, \textit{Theor. Probab.
Appl.} \textbf{35}, 216--227, (1990).

\bibitem{8} P. M. Bleher,  \textit{Commun. Math. Phys.} \textbf{128},
411--419, (1990).

\bibitem{9} P. M. Bleher, J. Ruiz and V. A. Zagrebnov,
\textit{J.Stat.Phys.}, \textbf{79}, 473--482, (1995).

\bibitem{10} P. M. Bleher, J. Ruiz and V. A. Zagrebnov, \textit{J.Stat.Phys.}
\textbf{93}, 33--78, (1998).

\bibitem{11} D. Ioffe,  \textit{Lett. Math. Phys.} \textbf{37}, 137--143, (1996).

\bibitem{12} D. Ioffe, \textit{Extremality of the disordered state for the
Ising model on general trees,} Prog. Probab. Vol. \textbf{40}, pp.
3--14, (1995).

\bibitem{13} P. M. Bleher, J. Ruiz, R. H. Schonmann, S. Shlosman and V. A.
Zagrebnov,  \textit{Moscow Math. J.} \textbf{3}, 345--363, (2001).

\bibitem{15} H. Saygili,  \textit{Asian Journal of Current Research} V. 1,
N 3, 114--121, (2017).

\bibitem{16} M. M. Rahmatullaev,  \textit{Russian Mathematics}, \textbf{59}:11, 45--53, (2015).

\bibitem{17} M. M. Rahmatullaev, \textit{Uzb. Mat. Zh.,} \textbf{2},
144--152, (2009).

\bibitem{18} C. Kulske, U. A. Rozikov, R. M. Khakimov. \textit{Description of the Translation-Invariant Splitting Gibbs
Measures for the Potts Model on a Cayley Tree. Jour. Stat. Phys.}
\textbf{156}:1, 189--200, (2014).

\bibitem{19} U.A.Rozikov and R.M.Khakimov,  \textit{Theor. Math.Phys.,}
\textbf{175}, 699--709, (2013).
\bibitem{20} R.M.Khakimov,  \textit{Uzb. Mat. Zh.,} \textbf{3}, 134--142, (2014).

\bibitem{21} M.M.Rahmatullaev,  \textit{Theor. Math. Phys.,} \textbf{180},
1019--1029, (2014).

\bibitem{22} M.M.Rahmatullaev,  \textit{Journal of Mathematical Physics, Analysis, Geometry,} vol. 12, N. 4, 302--314, (2016).

\bibitem{23} M.M.Rahmatullaev,  \textit{Ukrainian Mathematical Journal,} Vol. 68, N. 4, 598--611, (2016).

\bibitem{25} M. A. Rasulova, \textit{Theor. Math. Phys}.
\textbf{199}(1), 586--592, (2019).

\bibitem{26} U.A.Rozikov, Y.M.Suhov, \textit{Infinite Dimensional Analysis,
Quantum Probability and Related Topics.,} Vol.9, N. 3, 471--488,
(2006).

\bibitem{27} C. Kulske, U. A. Rozikov, \textit{Extremality of translation-invariant phases for a three-state SOS-model on the Binary tree, J.Stat.Phys.} \textbf{160}, 659--680,(2015).

\bibitem{28} H. Kesten, B. P. Stigum, \textit{Additional limit theorem for indecomposable multi-dimensional Galton-Watson processes, Ann.Math.Stat.} \textbf{37}, 1463--1481,(1966).

\bibitem{29} F. Martinelli, A. Sinclair, D. Weitz, \textit{Fast mixing for independent sets, coloring and other models on trees, Random Struct. Algoritms.} \textbf{31}, 134--172,(2007).



\end{thebibliography}
\end{document}